\documentclass{amsart}
\usepackage{amsfonts}
\usepackage{amsmath}
\usepackage{amssymb}
\usepackage{graphicx}

\theoremstyle{plain}
\newtheorem{theorem}{Theorem}[section]
\newtheorem{lemma}[theorem]{Lemma}
\newtheorem{proposition}[theorem]{Proposition}
\theoremstyle{definition}

\theoremstyle{remark}
\newtheorem{remark}[theorem]{Remark}
\numberwithin{equation}{section}

\begin{document}
\title[On the integrability of the Ostrovsky-Vakhnenko equation]{On the
complete integrability of the Ostrovsky-Vakhnenko equation}
\author{Yarema A. Prykarpatskyy}
\address{Department of Applied Mathematics, University of Agriculture, ul.
Balicka 253c, 30-198 Krakow, Poland, and the Institute of Mathematics of NAS
of Ukraine, 3 Tereshchenkivska str., Kyiv, Ukraine}
\email{yarpry@gmail.com}
\date{present}

\begin{abstract}
The complete integrability of the Ostrovsky-Vakhnenko equation is studied by
means of symplectic gradient-holonomic and differential-algebraic tools. A
compatible pair of polynomial Poissonian structures, Lax type representation
and related infinite hierarchies of conservation laws are constructed.
\end{abstract}

\maketitle

\section{Introduction}

In 1998 V.O.~Vakhnenko investigated high-frequency perturbations in a
relaxing barothropic medium. He discovered that this phenomenon is described
by a new nonlinear evolution equation. Later it was proved that this
equation is equivalent to the reduced Ostrovsky equation \cite{Os}, which
describes long internal waves in a rotating ocean. The nonlinear
integro-differential Ostrovsky-Vakhnenko equation
\begin{equation}
u_{t}=-uu_{x}-D_{x}^{-1}u  \label{D1.1}
\end{equation}%
on the real axis $\mathbb{R}$ for a smooth function $u\in C^{(\infty )}(%
\mathbb{R}^{2};\mathbb{R)},$ where $D_{x}^{-1}$ is the inverse-differential
operator to $D_{x}:=\partial /\partial x,$ can be derived \cite{Va} as a
special case of the Whitham type equation
\begin{equation}
u_{t}=-uu_{x}+\int_{\mathbb{R}}K(x,y)u_{y}dy.  \label{D1.2}
\end{equation}%
Here the generalized kernel $K(x,y):=\frac{1}{2}|x-y|,x,y\in \mathbb{R}$ and$%
\ t\in \mathbb{R}$ is an evolution parameter. Different analytical
properties of equation (\ref{D1.1}) were analyzed in articles \cite%
{Os,Va,Par}, the corresponding Lax type integrability was stated in \cite%
{DHH}.

Recently by J.C. Brunelli and S. Sakovich in \cite{BS} there was
demonstrated that Ostrovsky-Vakhnenko equation is a suitable reduction of
the well known Camassa-Holm equation that made it possible to construct the
corresponding compatible Poisson structures for (\ref{D1.1}), but in a
complicated enough non-polynomial form.

In the present work we will reanalyze the integrability of equation (\ref%
{D1.1}) both from the gradient-holonomic \cite{PM,BPS,HPP}, symplectic and
formal differential-algebraic points of view. As a result, we will re-derive
the Lax type representation for the Ostrovsky-Vakhnenko equation (\ref{D1.1}%
), construct the related simple enough compatible polynomial Poisson
structures and an infinite hierarchy of conservation laws.

\section{Gradient-holonomic integrability analysis}

Consider the nonlinear Ostrovsky-Vakhnenko equation (\ref{D1.1}) as a
suitable nonlinear dynamical system
\begin{equation}
du/dt=-uu_{x}-D_{x}^{-1}u:=K[u]  \label{D2.1}
\end{equation}%
on the smooth $2\pi $-periodic functional manifold
\begin{equation}
M:=\{u\in C^{(\infty )}(\mathbb{R}/2\pi \mathbb{Z};\mathbb{R)}%
:\int_{0}^{2\pi }udx=0\},  \label{D2.2}
\end{equation}%
where $K:M\rightarrow T(M)$ is the corresponding well-defined smooth vector
field on $M.$

We, first, will state that the dynamical system (\ref{D2.1}) on manifold $M $
possesses an infinite hierarchy of conservation laws, that can signify as a
necessary condition for its integrability. For this we need to construct a
solution to the Lax gradient equation%
\begin{equation}
\varphi _{t}+K^{^{\prime ,\ast }}\varphi =0,  \label{D2.3}
\end{equation}%
in the special asymptotic form%
\begin{equation}
\varphi =\exp [-\lambda t+D_{x}^{-1}\sigma (x;\lambda )],  \label{D2.4}
\end{equation}%
where, by definition, a linear operator $K^{^{\prime ,\ast }}$\ $:T^{\ast
}(M)\rightarrow T^{\ast }(M)$ is, \ adjoint with respect to the standard
convolution $(\cdot ,\cdot )$ on $T^{\ast }(M)\times T(M),$ \bigskip the
Frechet-derivative of a nonlinear mapping $K:M\rightarrow T(M):$%
\begin{equation}
K^{^{\prime ,\ast }}=uD_{x}+D_{x}^{-1}  \label{D2.5}
\end{equation}%
and, respectively,
\begin{equation}
\sigma (x;\lambda )\simeq \sum_{j\in \mathbb{Z}_{+}}\sigma _{j}[u]\lambda
^{-j},  \label{D2.6}
\end{equation}%
as $|\lambda |\rightarrow \infty $ with some "local" functionals $\sigma
_{j}:M\rightarrow C^{(\infty )}(\mathbb{R}/2\pi \mathbb{Z};\mathbb{R)}$ \ on
$M$ for all $j\in \mathbb{Z}_{+}.$

By substituting (\ref{D2.4}) into (\ref{D2.3}) one easily obtains the
following recurrent sequence of functional relationships:%
\begin{equation}
\sigma _{j,t}+\sum_{k\leq j}\sigma _{j-k}(u\sigma _{k}+D_{x}^{-1}\sigma
_{k,t})-\sigma _{j+1}+(u\sigma _{j})_{x}+\delta _{j,0}=0  \label{D2.7}
\end{equation}%
for all $j+1\in \mathbb{Z}_{+}$ modulo the equation \ (\ref{D2.1}). By means
of standard calculations one obtains that this recurrent sequence is
solvable and
\begin{eqnarray}
\sigma _{0}[u] &=&0,\sigma _{1}[u]=1,\sigma _{2}[u]=u_{x},  \label{D2.8} \\
\sigma _{3}[u] &=&0,\sigma _{4}[u]=u_{t}+2uu_{x},  \notag \\
\sigma _{5}[u] &=&3/2(u^{2})_{xt}+u_{tt}+2/3(u^{3})_{xx}-u_{x}D_{x}^{-1}u
\notag
\end{eqnarray}%
and so on. It is easy check that all of functionals%
\begin{equation}
\gamma _{j}:=\int_{0}^{2\pi }\sigma _{j}[u]dx  \label{D2.9}
\end{equation}%
are on the manifold $M$ conservation laws, that is $d\gamma _{j}/dt=0$ for $%
j\in \mathbb{Z}_{+}$ with respect to the dynamical system \ (\ref{D2.1}). \
For instance, at $j=5$ one obtains:
\begin{eqnarray}
\gamma _{5} &:&=\int_{0}^{2\pi }\sigma _{5}[u]dx=\int_{0}^{2\pi }\left[
3/2(u^{2})_{xt}+u_{tt}+2/3(u^{3})_{xx}-u_{x}D_{x}^{-1}u\right] dx=
\label{D2.10} \\
&=&\frac{d^{2}}{dt^{2}}\int_{0}^{2\pi }u_{tt}dx-\int_{0}^{2\pi
}u_{x}D_{x}^{-1}udx=\frac{d^{2}}{dt^{2}}\int_{0}^{2\pi }udx-\left.
uD_{x}^{-1}u\right\vert _{0}^{2\pi }+\int_{0}^{2\pi }u^{2}dx=  \notag \\
&=&\int_{0}^{2\pi }u^{2}dx,  \notag
\end{eqnarray}%
and
\begin{eqnarray}
d\gamma _{5}/dt &=&2\int_{0}^{2\pi }uu_{t}dx=-2\int_{0}^{2\pi
}u(uu_{x}+D_{x}^{-1}u)dx=  \label{D2.10a} \\
&=&-2\int_{0}^{2\pi }uD_{x}^{-1}udx=-\int_{0}^{2\pi
}[(D_{x}^{-1}u)^{2}]_{x}dx=\left. (D_{x}^{-1}u)^{2}\right\vert _{0}^{2\pi
}=0,  \notag
\end{eqnarray}%
since owing to the constraint (\ref{D2.2}) the integrals $\left.
(D_{x}^{-1}u)\right\vert _{0}^{2\pi }=0.$

The result stated above allows us to suggest that the dynamical system (\ref%
{D2.1}) on the functional manifold $M$ is an integrable Hamiltonian system.

First, we will show that this dynamical system is a Hamiltonian flow
\begin{equation}
du/dt=-\vartheta \text{ }{grad}\ H[u]  \label{D2.11}
\end{equation}%
with respect to some Poisson structure $\vartheta :T^{\ast }(M)\rightarrow
T(M)$ and a Hamiltonian function $H\in \mathcal{D}(M).$ Based on the
standard symplectic techniques \cite{FT,PM,Bl,BPS} consider the conservation
law \ (\ref{D2.10}) and present it in the scalar "momentum" form:%
\begin{equation}
-1/2\gamma _{5}=\frac{1}{2}\int_{0}^{2\pi
}u_{x}D_{x}^{-1}udx=(1/2D_{x}^{-1}u,u_{x}):=(\psi ,u_{x})  \label{D2.12}
\end{equation}%
with the co-vector $\psi :=1/2D_{x}^{-1}u\in T^{\ast }(M)$ and calculate the
corresponding co-Poissonian structure%
\begin{equation}
\vartheta ^{-1}:=\psi ^{\prime }-\psi ^{\prime ,\ast }=D_{x}^{-1},
\label{D2.13}
\end{equation}%
or the Poissonian structure
\begin{equation}
\vartheta =D_{x}.  \label{D2.13a}
\end{equation}%
The obtained operator $\vartheta =D_{x}:T^{\ast }(M)\rightarrow T(M)$ is
really Poissonian for (\ref{D2.1}) since the following determining
symplectic condition
\begin{equation}
\psi _{t}+K^{^{\prime ,\ast }}\psi ={grad}\ \mathcal{L}  \label{D2.14}
\end{equation}%
holds for the Lagrangian function%
\begin{equation}
\mathcal{L=}\frac{1}{12}\int_{0}^{2\pi }u^{3}dx.  \label{D2.15}
\end{equation}%
As a result of (\ref{D2.14}) one obtains easily that
\begin{equation}
du/dt=-\vartheta {grad}\ H[u],  \label{D2.16}
\end{equation}%
where the Hamiltonian function
\begin{equation}
H=(\psi ,K)-\mathcal{L=}\frac{1}{2}\int_{0}^{2\pi
}[u^{3}/3-(D_{x}^{-1}u)^{2}/2]dx  \label{D2.17}
\end{equation}%
is an additional conservation law of the dynamical system (\ref{D2.1}).
Thus, one can formulate the following proposition.

\begin{proposition}
\label{Prop_D2.1} The Ostrovsky-Vakhnenko dynamical system (\ref{D2.1})
possesses an infinite hierarchy of nonlocal, in general, conservation laws (%
\ref{D2.9}) and is a Hamiltonian flow (\ref{D2.16}) on the manifold $M$ with
respect to the Poissonian structure (\ref{D2.13a}).
\end{proposition}

\begin{remark}
It is useful to remark here that the existence of an infinite ordered by $%
\lambda $-powers hierarchy of conservations laws (\ref{D2.9}) is a typical
property \cite{FT,Bl,PM,BPS} of the Lax type integrable Hamiltonian systems,
which are simultaneously bi-Hamiltonian flows with respect to corresponding
two compatible Poissonian structures.
\end{remark}

As is well known \cite{FT,Bl,PM,BPS}, the second Poissonian structure $\eta
:T^{\ast }(M)\rightarrow T(M)$ on the manifold $M$ for (\ref{D2.1}), if it
exists, can be calculated as
\begin{equation}
\eta ^{-1}:=\tilde{\psi}^{\prime }-\tilde{\psi}^{\prime ,\ast },
\label{D2.18}
\end{equation}%
where a-covector $\tilde{\psi}\in T^{\ast }(M)$ is a second solution to the
determining equation (\ref{D2.14}):%
\begin{equation}
\tilde{\psi}_{t}+K^{^{\prime ,\ast }}\tilde{\psi}={grad}\ \mathcal{\tilde{L}}
\label{D2.19}
\end{equation}%
for some Lagrangian functional $\mathcal{\tilde{L}}\in \mathcal{D}(M).$ It
can be certainly done by means of simple enough but cumbersome analytical
calculations based, for example, on the asymptotical small parameter method
\cite{PM,BPS,HPP} and on which we will not stop here.

Instead of this we will shall apply the direct differential-algebraic
approach to dynamical system (\ref{D2.1}) and reveal its Lax type
representation both in the differential scalar and in canonical matrix
Zakharov-Shabat forms. Moreover, we will construct the naturally related
compatible polynomial Poissonian structures for Ostrovsky -Vakhnenko
dynamical system (\ref{D2.1}) and generate an infinite hierarchy of
commuting to each other nonlocal conservation laws.

\section{Lax type representation and compatible Poissonian structures - the
differential-algebraic approach}

We will start with construction of the polynomial differential ring $%
\mathcal{K}\{u\}\subset \mathcal{K}:=\mathbb{R}\{\{x,t\}\}$ generated by a
fixed functional variable $u\in \mathbb{R}\{\{x,t\}\}$ and invariant with
respect to two differentiations $D_{x}:=\partial /\partial x$ and $%
D_{t}:=\partial /\partial t+u\partial /\partial x,$ satisfying the
Lie-algebraic commutator relationship

\begin{equation}
\lbrack D_{x},D_{t}]=u_{x}D_{x}.  \label{D3.1}
\end{equation}%
Since the Lax type representation for the dynamical system (\ref{D2.1}) can
be interpreted \cite{PAPP,BPS} as the existence of a finite-dimensional
invariant differential ideal $\mathcal{I}_{N}\{u\}\subset \mathcal{K}\{u\},\
$realizing the corresponding finite-dimensional representation of the the
Lie-algebraic commutator relationship (\ref{D3.1}), this ideal can be
presented as
\begin{equation}
\mathcal{I}_{N}\{u\}:=\{\sum_{j=\overline{0,N}}g_{j}D_{x}^{j}f[u]\in
\mathcal{K}\{u\}:g_{j}\in \mathcal{K},j=\overline{0,N}\},  \label{D3.2}
\end{equation}%
where an element $f[u]\in \mathcal{K}\{u\}$ and $N\in \mathbb{Z}_{+}$are
fixed. The $D_{x}$-invariance of ideal (\ref{D3.2}) will be \textit{a priori
}evident, if the function $f[u]\in \mathcal{K}\{u\}$ satisfies the linear
differential relationship%
\begin{equation}
D_{x}^{N+1}f=\sum_{k=0}^{N}a_{j}[u]D_{x}^{j}f  \label{D3.2a}
\end{equation}%
for some coefficients $a_{j}[u]\in \mathcal{K}\{u\},$ $j=\overline{0,N},$ \
but its $D_{t}$-invariance strongly depends on the element $f[u]\in \mathcal{%
K}\{u\},$ which can be found from the functional relationship \ (\ref{D2.3})
on the element $\varphi \lbrack u;\lambda ]:={grad}\ \gamma (\lambda )\in
\mathcal{K}\{u\},\gamma (\lambda ):=\int_{0}^{2\pi }\sigma (x;\lambda )dx,$
rewritten in the following form:%
\begin{equation}
D_{x}D_{t}\varphi =-\varphi .  \label{D3.3}
\end{equation}%
From the right hand side one follows that there exists an element $\eta
:=\eta \lbrack u]=-D_{t}\varphi \lbrack u]\in \mathcal{K}\{u\},$ such that
\begin{equation}
\varphi =D_{x}\eta .  \label{D3.3a}
\end{equation}%
Having substituted \ (\ref{D3.3a}) into the left hand side of \ (\ref{D3.3})
one finds easily that%
\begin{equation}
\begin{array}{c}
D_{x}D_{t}\eta -u_{x}\eta _{x}=D_{x}D_{t}\eta -u_{x}\varphi = \\
=D_{x}D_{t}\eta -u_{x}{grad}\ \gamma (\lambda )= \\
=D_{x}(D_{t}\eta -\gamma \lbrack u,\lambda ])\ =\eta ,%
\end{array}
\label{D3.3b}
\end{equation}%
where we have put, by definition, $\ \gamma (\lambda ):=$ $\int_{0}^{2\pi
}\gamma \lbrack u;\lambda ]dx$ for a suitably \ chosen density element $%
\gamma \lbrack u;\lambda ]\in \mathcal{K}\{u\}.$ \ As an evident result of \
(\ref{D3.3b}) one derives that there exists an element $\rho :=\rho \lbrack
u]\in \mathcal{K}\{u\},$ such that
\begin{equation}
\eta =D_{x}\rho .  \label{D3.3c}
\end{equation}%
Turning back to the relationships \ (\ref{D3.3a}) and \ (\ref{D3.3c}) one
obtains that the following differential representation
\begin{equation}
\varphi =D_{x}^{2}\rho  \label{D3.3d}
\end{equation}%
holds.

As a further step, we \ can try to realize \ the differential ideal \ (\ref%
{D3.2}) \ by means of the generating element $f[u]\Longrightarrow \rho
\lbrack u]\in \mathcal{K}\{u\},\ $defined by the relationship \ (\ref{D3.3d}%
). But, as it is easy to check, the obtained this way differential ideal is
not finite-dimensional. So, for a future calculating convenience, we will
represent the element $\rho \lbrack u]\in \mathcal{K}\{u\}$ in the following
natural factorized form:
\begin{equation}
\rho :=\bar{f}f,  \label{D3.4}
\end{equation}%
where elements $f,\bar{f}\in \mathcal{K}\{u\}$ satisfy, by definition, the
adjoint pairs of the following differential relationships:%
\begin{eqnarray}
D_{x}^{N+1}f &=&\sum_{k=0}^{N}a_{j}[u]D_{x}^{j}f,  \label{D3.4a} \\
(-1)^{N+1}D_{x}^{N+1}\bar{f} &=&\sum_{k=0}^{N}(-1)^{j}(D_{x}^{j}a_{j}[u])%
\bar{f},\text{ }  \notag
\end{eqnarray}%
and
\begin{equation}
D_{t}f=\sum_{j=0}^{N-1}b_{j}D_{x}^{j}f,\text{ \ \ }D_{t}\bar{f}=-u_{x}\bar{f}%
+\sum_{j=0}^{N-1}(-1)^{j+1}(D_{x}^{j}b_{j})f,  \label{D3.5}
\end{equation}%
for some elements $b_{j}\in \mathcal{K}\{u\},j=\overline{0,N-1},$ and check
the finite-dimensional $D_{x}$- and $D_{t}$-invariance of the corresponding
ideal \ (\ref{D3.2}), generated by the element $f\in \mathcal{K}\{u\}.$

Now it is easy to check by means of simple enough calculations, based on the
relationship \ (\ref{D3.3}) and (\ref{D3.3d}), that the following
differential equalities
\begin{eqnarray}
D_{x}(D_{t}\varphi ) &=&-\varphi ,\text{ \ \ \ \ \ }D_{x}(D_{t}^{2}\varphi
)=-u_{x}\varphi -D_{t}\varphi ,\   \label{D3.6} \\
D_{x}(D_{t}^{3}\varphi ) &=&u\varphi -2u_{x}D_{t}\varphi -D_{t}^{2}\varphi
,\   \notag \\
D_{x}(D_{t}^{4}\varphi ) &=&-(uu_{x}+D_{x}^{-1}u)\varphi
+(4u_{x}^{2}+3u)D_{t}\varphi -2u_{x}D_{t}^{2}\varphi -D_{t}^{3}\varphi ,...,
\notag
\end{eqnarray}%
and their consequences
\begin{eqnarray}
D_{t}D_{x}^{2}\rho &=&-\rho _{x},\text{ \ \ \ \ \ }D_{x}(D_{t}\rho
_{x})=u_{x}\rho _{xx}-D_{x}\rho ,\   \label{D3.6a} \\
D_{x}^{2}(D_{t}\rho ) &=&D_{x}(u_{x}D_{x}\rho -\rho )+u_{xx}D_{x}^{2}\rho
,...,\   \notag
\end{eqnarray}%
\ hold. \ Taking into account the independence of the sets of functional
elements $\{f,D_{x}f,D_{x}f,..,D_{x}^{N-1}f\}\subset \mathcal{K}\{u\}$ and $%
\{\bar{f},D_{x}\bar{f},D_{x}\bar{f},..,D_{x}^{N-1}\bar{f}\}\subset \mathcal{K%
}\{u\},$ the relationships \ \ (\ref{D3.6a}) jointly with (\ref{D3.4}), \ (%
\ref{D3.4a}) and \ (\ref{D3.5}) make it possible to state the following
lemma.

\begin{lemma}
\label{Lm_D3.1}\bigskip The set \ (\ref{D3.2}) represents a $D_{x}$- and $%
D_{t}$-invariant differential ideal in the ring $\ \mathcal{K}$ for all $%
N\geq 2.$
\end{lemma}

\begin{proof}
This result easily follows from the fact that for number $N\geq 2$ all of
the relationships \ (\ref{D3.6a}) persist to be compatible upon taking into
account the differential expressions (\ref{D3.4}) and \ (\ref{D3.5}).
Contrary to that, at $N=1$ they become not compatible.
\end{proof}

As a corollary of Lemma \ \ref{Lm_D3.1}, having put in \ (\ref{D3.2}) and \ (%
\ref{D3.5}) the number $N=2,$ one finds easily by means of elementary enough
calculations that the related differential ideal $\mathcal{I}_{2}\{u\}$
lasts to be invariant, if the differential Lax type relationships
\begin{equation}
D_{x}^{3}f=-\mu \bar{u}f,\text{ }D_{x}^{3}\bar{f}=\mu \bar{u}\bar{f},
\label{D3.7}
\end{equation}%
and
\begin{equation}
D_{t}f=\mu ^{-1}D_{x}^{2}f+u_{x}f,\text{ }D_{t}\bar{f}=-\mu ^{-1}D_{x}^{2}%
\bar{f}-2u_{x}\bar{f},  \label{D3.8}
\end{equation}%
where $\bar{u}:=u_{xx}+1/3,\mu \in \mathbb{C}\backslash \{0\}$ is an
arbitrary complex parameter, hold. Moreover, they exactly coincide with
those found before in \cite{DHH}. \ The obtained above differential
relationships \ (\ref{D3.7}) and \ (\ref{D3.8}) can be equivalently
rewritten in the following matrix Zakharov-Shabat type form:
\begin{equation}
D_{t}h=\hat{q}[u;\mu ]h,\text{ }D_{x}h=\hat{l}[u;\mu ]h,  \label{D3.8a}
\end{equation}%
where matrices

\begin{equation}
\hat{q}[u;\mu ]:=\left(
\begin{array}{ccc}
u_{x} & 0 & 1/\mu \\
-1/3 & 0 & 0 \\
0 & -1/3 & -u_{x}%
\end{array}%
\right) ,\text{ \ \ }\hat{l}[u;\mu ]:=\left(
\begin{array}{ccc}
0 & 1 & 0 \\
0 & 0 & 1 \\
-\mu \bar{u} & 0 & 0%
\end{array}%
\right)  \label{D3.8b}
\end{equation}%
and $h:=(f,D_{x}f,D_{x}^{2}f)^{\intercal }\in \mathcal{K}\{u\}^{3}.$

\ Based further on the obtained differential relationships \ (\ref{D3.7})
and \ (\ref{D3.8}), one obtains that the compatibility condition (\ref{D3.3}%
) gives rise to the following important relationship%
\begin{equation}
-\vartheta \varphi =D_{x}^{2}D_{t}\varphi =3\mu ^{2}\eta \varphi ,
\label{D3.9}
\end{equation}%
where \ the polynomial integro-differential operator%
\begin{equation}
\eta :=\partial ^{-1}\bar{u}\partial ^{-3}\bar{u}\partial ^{-1}+4\partial
^{-2}\bar{u}\partial ^{-1}\bar{u}\partial ^{-2}+2(\partial ^{-2}\bar{u}%
\partial ^{-2}\bar{u}\partial ^{-1}+\partial ^{-1}\bar{u}\partial ^{-2}\bar{u%
}\partial ^{-2})  \label{D3.10}
\end{equation}%
is skewsymmetric on the functional manifold $M$ and presents the second
compatible Poisson structure for the Ostrovsky-Vakhnenko dynamical system (%
\ref{D2.1}).

Based now on the recurrent relationships following from substitution of the
asymptotic expansion%
\begin{equation}
\varphi \simeq \sum_{j\in \mathbb{Z}_{+}}\varphi _{j}\xi ^{-j},\text{ }\xi
:=-1/(3\mu ^{2}),  \label{D3.13}
\end{equation}%
into (\ref{D3.9}), one can determine a new infinite hierarchy of
conservations laws for dynamical system (\ref{D2.1}):%
\begin{equation}
\tilde{\gamma}_{j}:=\int_{0}^{1}ds(\varphi _{j}[us],u),  \label{D3.14}
\end{equation}%
for $j\in \mathbb{Z}_{+},$ where
\begin{equation}
\varphi _{j}=\Lambda ^{j}\varphi _{0},\text{ }\vartheta \varphi _{0}=0,
\label{D3.15}
\end{equation}%
and the recursion operator $\Lambda :=\vartheta ^{-1}\eta :T^{\ast
}(M)\rightarrow T^{\ast }(M)$ satisfies the standard Lax type representation:%
\begin{equation}
\Lambda _{t}=[\Lambda ,K^{^{\prime }\ast }].  \label{D3.16}
\end{equation}%
The obtained above results can be formulated as follows.

\begin{proposition}
\label{Prop_D3.2} The Ostrovsky-Vakhnenko dynamical system (\ref{D2.1})
allows the standard differential Lax type representation (\ref{D3.7}), (\ref%
{D3.8}) and defines on the functional manifold $M$ an integrable
bi-Hamiltonian flow with compatible Poisson structures (\ref{D2.13a}) and (%
\ref{D3.10}). In particular, this dynamical system possesses an infinite
hierarchy of nonlocal conservation laws (\ref{D3.14}), defined by the
gradient elements (\ref{D3.15}).
\end{proposition}

\begin{remark}
It is useful to remark here that the existence of an infinite \ $\lambda $%
-powers ordered hierarchy of conservations laws (\ref{D2.9}) is a typical
property \cite{FT,Bl,PM,BPS} of the Lax type integrable Hamiltonian systems,
which are simultaneously bi-Hamiltonian flows with respect to corresponding
compatible Poissonian structures.
\end{remark}

\begin{remark}
It is interesting to observe that our second polynomial Poisson structure (%
\ref{D3.10}) differs from that obtained recently in \cite{BS}, which
contains the rational power factors.
\end{remark}

It is easy to construct making use of the differential expressions (\ref%
{D3.7}) and (\ref{D3.8}) a slightly different from \ (\ref{D3.8a}) matrix
Lax type representation of the Zakharov-Shabat form for the dynamical system
(\ref{D1.1}).

Really, if to define the "spectral" parameter $\mu :=1/(9\lambda )\in $ $%
\mathbb{C}\backslash \{0\}$ and new basis elements of the invariant
differential ideal \ (\ref{D3.2}):%
\begin{equation}
g_{1}:=-3D_{x}f,\text{ }g_{2}:=f,\text{ }g_{3}:=9\lambda D_{x}^{2}f+u_{x}f,
\label{D3.17}
\end{equation}%
then relationships (\ref{D3.7}) and (\ref{D3.8}) can be rewritten as follows:%
\begin{equation}
D_{t}g=q[u;\lambda ]g,\text{ }D_{x}g=l[u;\lambda ]g,  \label{D3.18}
\end{equation}%
where matrices%
\begin{equation}
q[u;\lambda ]:=\left(
\begin{array}{ccc}
0 & 1 & 0 \\
0 & 0 & 1 \\
\lambda & -u & 0%
\end{array}%
\right) ,l[u;\lambda ]:=\left(
\begin{array}{ccc}
0 & u_{x}/(3\lambda ) & -1/(3\lambda ) \\
-1/3 & 0 & 0 \\
-u_{x}/3 & -1/3 & 0%
\end{array}%
\right)  \label{D3.19}
\end{equation}%
coincide with those of \cite{DHH,BS} and satisfy the following
Zakharov-Shabat type compatibility condition:
\begin{equation}
D_{t}l=[q,l]+D_{x}q-l\text{ }D_{x}u.  \label{D3.20}
\end{equation}

\begin{remark}
As it was already mentioned above, the Lax type representation \ (\ref{D3.19}%
) of the  Ostrovsky-Vakhnenko dynamical system \ (\ref{D1.1}) was obtained
in  \cite{DHH} by means of a suitable limiting reduction of the
Degasperis-Processi equation
\begin{equation}
u_{t}-u_{xxt}+4uu_{x}-3u_{x}u_{xx}-uu_{xxx}=0.  \label{D3.21}
\end{equation}%
For convenience, let us rewrite the latter in the following form:%
\begin{equation}
D_{t}z=-3zD_{x}u,\text{ }\ z=u-D_{x}^{2}u,  \label{D3.22}
\end{equation}%
where differentiations $D_{t}:=\partial /\partial t+u\partial /\partial x$
and $D_{x}:=\partial /\partial x$ \ satisfy the Lie- algebraic relationship
\ (\ref{D3.1}). It appears to be very impressive that equation \ (\ref{D3.21}%
) is itself a  special reduction of a new  Lax type integrable Riemann type
hydrodynamic system,  proposed and studied (at $s=2)$ recently  in  \cite%
{BPAP}:%
\begin{equation}
D_{t}^{N-1}u=\bar{z}_{x}^{s}\ ,\text{ \ \ }D_{t}\bar{z}=0,  \label{D3.23}
\end{equation}
where $s,$ $N\in \mathbb{N}$ are arbitrary natural numbers.  Really, having
put, by definition,  $z:=\bar{z}_{x}^{s}$ and $s=3,$ from \ (\ref{D3.23})
one easily obtains the following dynamical system:%
\begin{equation}
\begin{array}{c}
D_{t}^{N-1}u=z, \\
D_{t}z=-3zD_{x}u,%
\end{array}
\label{D3.24}
\end{equation}%
coinciding with the   Degasperis-Processi equation \ (\ref{D3.22}) if to
make the identification  $z=u-D_{x}^{2}u.$ As a result, we have stated that
a function $u\in C^{\infty }(\mathbb{R}^{2};\mathbb{R})$, satisfying for an
arbitrary $N\in \mathbb{N}$ the generalized Riemann type hydrodynamical
equation
\begin{equation}
D_{t}^{N-1}u=u-D_{x}^{2}u,  \label{D3.25}
\end{equation}%
simultaneously solves the Degasperis-Processi equation \ (\ref{D3.21}). In
particular, having put $N=2,$ we obtain that solutions to the Burgers type
equation
\begin{equation}
D_{t}u=u-D_{x}^{2}u  \label{D3.26}
\end{equation}%
are solving also the Degasperis-Processi equation \ (\ref{D3.21}). It means,
in particular, that the reduction procedure of the work  \cite{DHH} can be
also applied to the Lax type integrable Riemann type hydrodynamic system \ (%
\ref{D3.23}), giving rise to a related Lax type representation for the
Ostrovsky-Vakhnenko dynamical system \ (\ref{D1.1}).
\end{remark}

\bigskip

\section{Conclusion}

We have showed that the Ostrovsky-Vakhnenko dynamical system is naturally
embedded into the general Lax type integrability scheme \cite{FT,Bl,PM,BPS},
whose main ingredients such as the corresponding compatible Poissonian
structures and Lax type representation can be effectively enough retrieved
by means of direct modern integrability tools, such as the
differential-geometric, differential-algebraic and symplectic gradient
holonomic approaches. We have also demonstrated the relationship of the Ostrovsky-Vakhnenko equation \ref{D1.1} with a generalize Riemann type hydrodynamic system, studied recently in \cite{BPAP} and its reduction. 

\section{Acknowledgements}

Author acknowledges the Scientific and Technological Research Council of
Turkey (TUBITAK/NASU-111T558 Project) for a partial support of his research.

\end{document}